\documentclass{article}

\usepackage{microtype}
\usepackage{graphicx}
\usepackage{subfigure}
\usepackage{booktabs} % for professional tables

\usepackage{algorithm}
\usepackage{algorithmic}

\usepackage{amsfonts}
\usepackage{amsmath}
\usepackage{amsthm}
\usepackage{amssymb}
\usepackage{epsf}
\usepackage{amssymb,amsmath}
\usepackage{times,bm,mathrsfs}
\usepackage{mathrsfs}
\usepackage{bbm}
\usepackage{xcolor}
\usepackage{colordvi}
\usepackage{amscd}
\usepackage{epsfig}
\usepackage{times}
\usepackage{balance}
\usepackage{comment}

\newcommand{\mission}{\textbf{\textsc{mission}}}
\usepackage{hyperref}

\newtheorem{theorem}{Theorem}

\newtheorem{lemma}{Lemma}

\newtheorem{assumption}{Assumption}

\usepackage[accepted]{icml2018}

\icmltitlerunning{MISSION: Feature Selection via Sketching}

\begin{document}

\twocolumn[
\icmltitle{MISSION: Ultra Large-Scale Feature Selection using Count-Sketches}
\icmlsetsymbol{equal}{*}

\begin{icmlauthorlist}
\icmlauthor{Amirali Aghazadeh}{equal,stanford}
\icmlauthor{Ryan Spring}{equal,rice-cs}
\icmlauthor{Daniel LeJeune}{rice-ece}
\icmlauthor{Gautam Dasarathy}{rice-ece}
\icmlauthor{Anshumali Shrivastava}{rice-cs}
\icmlauthor{Richard G. Baraniuk}{rice-ece}
\end{icmlauthorlist}

\icmlaffiliation{stanford}{Department of Electrical Engineering, Stanford University, Stanford, California}
\icmlaffiliation{rice-cs}{Department of Computer Science, Rice University, Houston, Texas}
\icmlaffiliation{rice-ece}{Department of Electrical and Computer Engineering, Rice University, Houston, Texas}
\icmlcorrespondingauthor{Anshumali Shrivastava}{anshumali@rice.edu}
\icmlkeywords{Variable Selection; Feature Selection; Large-Scale Machine Learning; Randomized Hashing}

\vskip 0.3in
]

\printAffiliationsAndNotice{\icmlEqualContribution}

\begin{abstract}
Feature selection is an important challenge in machine learning. 
It plays a crucial role in the {\em explainability} of machine-driven decisions that are rapidly permeating throughout modern society. 
Unfortunately, the explosion in the size and dimensionality of real-world datasets poses a severe challenge to standard feature selection algorithms. 
Today, it is not uncommon for datasets to have billions of dimensions. 
At such scale, even storing the feature vector is impossible, causing most existing feature selection methods to fail. 
Workarounds like feature hashing, a standard approach to large-scale machine learning, helps with the computational feasibility, but at the cost of losing the interpretability of features. 
In this paper, we present \mission{}, a novel framework for ultra large-scale feature selection that performs stochastic gradient descent while maintaining an efficient representation of the features in memory using a Count-Sketch data structure. 
\mission{} retains the simplicity of feature hashing without sacrificing the interpretability of the features while using only $\mathcal{O}(\log^2{p})$ working memory. 
We demonstrate that \mission{} accurately and efficiently performs feature selection on real-world, large-scale datasets with billions of dimensions.  
\end{abstract}

\section{Introduction}
\label{sec:introduction}

Feature selection is an important step in extracting interpretable patterns from data. It has numerous applications in a wide range of areas, including natural-language processing, genomics, and chemistry. Suppose that there are $n$ ordered pairs $({\bf X}_i,y_i)_{i\in [n]}$, where ${\bf X}_i \in \mathbb{R}^p$ are $p$-dimensional {\em feature vectors}, and $y_i\in \mathbb{R}$ are scalar outputs. Feature selection aims to identify a small subset of features (coordinates of the $p$-dimensional feature vector) that best models the relationship between the data ${\bf X}_i$ and the output $y_i$. 

A significant complication that is common in modern engineering and scientific applications is that the feature space $p$ is ultra high-dimensional. For example, Weinberger introduced a dataset with 16 trillion ($p=10^{13}$) unique features \cite{weinberger2009feature}. A 16 trillion dimensional feature vector (of double 8 bytes) requires 128 terabytes of working memory. Problems from modern genetics are even more challenging. A particularly useful way to represent a long DNA sequence is by a feature vector that counts the occurrence frequency of all length-$K$ sub-strings called $K$-mers. This representation plays an important role in large-scale regression problems in computational biology \cite{wood2014kraken,bray2015near,vervier2016large,aghazadeh2016universal}. Typically, $K$ is chosen to be  larger than 12, and these strings are composed of all possible combinations of 16 characters ($\{\text{A},\text{T},\text{C},\text{G}\}$ in addition to 12 wild card characters). In this case, the feature vector dimension is $p = 16^{12} = 2^{48}$. A vector of size $2^{48}$ single-precision variables requires approximately 1 petabyte of space! 

For ultra large-scale feature selection problems, it is impossible to run standard explicit regularization-based methods like $\ell_1$ regularization \cite{shalev2011stochastic,tan2014towards} or to select hyperparameters with a constrained amount of memory \cite{langford2009sparse}. This is not surprising, because these methods are not scalable in terms of memory and computational time \cite{duchi2008efficient}. Another important operational concern is that most datasets represent features in the form of strings or tokens. For example, with DNA or $n$-gram datasets, features are represented by strings of characters. Even in click-through data \cite{mcmahan2013ad}, features are indexed by textual tokens. Observe that mapping each of these strings to a vector component requires maintaining a dictionary whose size equals the length of the feature vector. As a result, one does not even have the capability to create a numerical exact vector representation of the features. 

Typically, when faced with such large machine learning tasks, the practitioner chooses to do \emph{feature hashing}~\cite{weinberger2009feature}. Consider a 3-gram string ``abc". With feature hashing, one uses a lossy, random hash function $h:{\rm strings} \rightarrow \{0,1,2,\dots,R\}$ to map ``abc" to a feature number $h(abc)$ in the range $\{0,1,2,\dots,R\}$. This is extremely convenient because it enables one to avoid creating a large look-up dictionary. Furthermore, this serves as a dimensionality reduction technique, reducing the problem dimension to $R$. Unfortunately, this convenience comes at a cost. Given that useful dimensionality reduction is strictly surjective (i.e., $R < p$), we lose the identity of the original features. This is not a viable option if one cares about both feature selection and interpretability.

One reason to remain hopeful is that in such high-dimensional problems, the data vectors ${\bf X}_i$ are extremely sparse~\cite{wood2014kraken}. For instance, the DNA sequence of an organism contains only a small fraction (at most the length of the DNA sequence) of $p = 16^{12}$ features. The situation is similar whether we are predicting click-through rates of users on a website or if we seek $n$-gram representations of text documents \cite{mikolov2013efficient}. In practice, ultra high-dimensional data is almost always ultra-sparse. Thus, loading a sparse data vector into memory is usually not a concern. The problem arises in the intermediate stages of traditional methods, where dense iterates need to be tracked in the main memory. One popular approach is to use \emph{greedy thresholding} methods~\cite{maleki2009coherence,mikolov2013efficient,jain2014iterative,jain2017partial} combined with stochastic gradient descent (SGD) to prevent the feature vector ${\boldsymbol \beta}$ from becoming too dense and blowing up in memory. In these methods, the intermediate iterates are regularized at each step, and a full gradient update is never stored nor computed (since this is memory and computation intensive). However, it is well known that greedy thresholding can be myopic and can result in poor convergence. We clearly observe this phenomenon in our evaluations. See Section~\ref{sec:simulation} for details.

In this paper we tackle the ultra large-scale feature selection problem, i.e., feature selection with billions or more dimensions. We propose a novel feature selection algorithm called \mission{}, a {\bf M}emory-efficient, {\bf I}terative {\bf S}ketching algorithm for {\bf S}parse feature select{\bf ION}. \mission{}, that takes on all the concerns outlined above. \mission{} matches the accuracy performance of existing large-scale machine learning frameworks like Vowpal Wabbit (VW) \cite{agarwal2014reliable} on real-world datasets. However, in contrast to VW, \emph{\mission{} can perform feature selection exceptionally well}. Furthermore, \mission{} significantly surpasses the performance of classical algorithms such as Iterative Hard Thresholding (IHT), which is currently the popular feature selection alternative concerning the problem sizes we consider. 

{\bf Contributions:} In this work, we show that the two-decade old Count-Sketch data structure~\cite{charikar2002finding} from the streaming algorithms literature is ideally suited for ultra large-scale feature selection. The Count-Sketch data structure enables us to retain the convenience of feature hashing along with the identity of important features. Moreover, Count-Sketch can accumulate gradients updates over several iterations because of linear aggregation. This aggregation eliminates the problem of myopia associated with existing greedy thresholding approaches.

In particular, we force the parameters (or feature vector) to reside in a memory-efficient Count-Sketch data structure \cite{charikar2002finding}. SGD gradient updates are easily applied to the Count-Sketch.
\begin{figure}[t]
\vspace{0cm}
\centering
\includegraphics[width=0.45\textwidth]{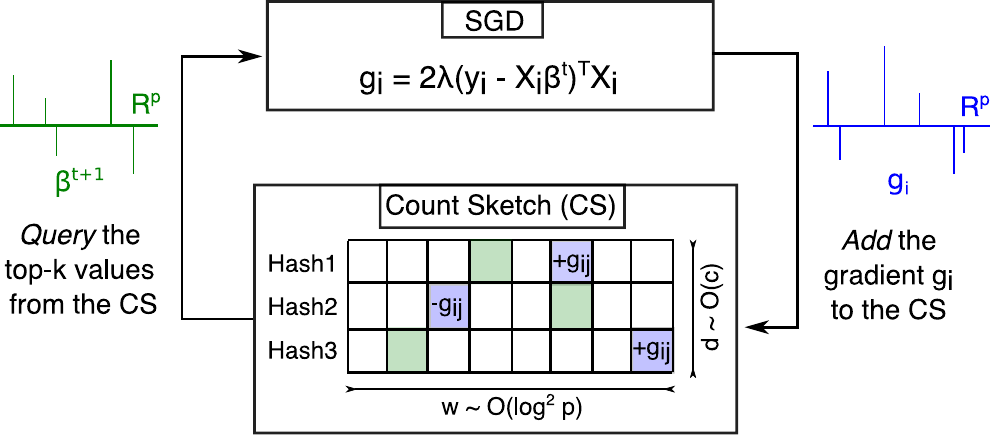}
\vspace{0cm}
\caption{Schematic of the \mission{} algorithm. \mission{} iteratively adds the stochastic gradient term $g_i \in \mathbb{R}^p$ into a Count-Sketch and queries back the top-$k$ heavy hitters from the Count-Sketch. The Count-Sketch requires $\mathcal{O}(\text{log}^2 p)$ memory to store a sketch of the $\mathcal{O}(p)$-dimensional feature vector ${\boldsymbol \beta}$. }
\label{fig:mission_schem}
\vspace{0cm}
\end{figure}
Instead of moving in the gradient direction and then greedily projecting into a subspace defined by the regularizer (e.g., in the case of LASSO-based methods), \mission{} adds the gradient directly into the Count-Sketch data structure, where it aggregates with all the past updates. See Fig.~\ref{fig:mission_schem} for the schematic. At any point of time in the iteration, this data structure stores a compressed, randomized, and noisy sketch of the sum of all the gradient updates, while preserving the information of the \emph{heavy-hitters}---the coordinates that accumulate the highest amount of energy. In order to find an estimate of the feature vector, \mission{} queries the Count-Sketch. The Count-Sketch is used in conjunction with a top-$k$ heap, which explicitly stores the features with the heaviest weights. Only the features in the top-$k$ heap are considered active, and the rest are set to zero. However, a representation for every weight is stored, in compressed form, inside the Count-Sketch.

We demonstrate that \mission{} surpasses the sparse recovery performance of classical algorithms such as Iterative Hard Thresholding (IHT), which is the only other method we could run at our scale. In addition, experiments suggest that the memory requirements of \mission{} scale well with the dimensionality $p$ of the problem. \mission{} matches the accuracy of existing large-scale machine learning frameworks like Vowpal Wabbit (VW) on real-world, large-scale datasets. Moreover, \mission{} achieves comparable or even better accuracy while using significantly fewer features. 

\section{Review: Streaming Setting and the Count-Sketch Algorithm}
\label{sec:review_CS}

\label{sec:Countsketch}
In the streaming setting, we are given a very high-dimensional vector ${\boldsymbol \beta} \in \mathbb{R}^p$ that is too costly to store in memory. We see only a very long sequence of updates over time. The only information available at time $t$ is of the form $(i, \Delta)$, which means that coordinate $i$ is incremented (or decremented) by the amount $\Delta$. We are given a limited amount of storage, on the order of $\mathcal{O}(\log{p})$, which means that we can never store the entire sequence of updates. Sketching algorithms aim to estimate the value of current item $i$, after any number of updates using only $\mathcal{O}(\log{p})$ memory. Accurate estimation of heavy coordinates is desirable.

Count-Sketch is a popular algorithm for estimation in the streaming setting. Count-Sketch keeps a matrix of counters (or bins) $\mathcal{S}$ of size $d \times w \sim \mathcal{O}(\log{p})$, where $d$ and $w$ are chosen based on the accuracy guarantees. The algorithm uses $d$ random hash functions $h_j \ \ j \in \{1,\ 2,...,\ d\}$ to map the vector's components to bins $w$. $h_j:\{1,\ 2,..., \ p\} \rightarrow \{1, \ 2,...,\ w\}$ Every component $i$ of the vector is hashed to $d$ different bins. In particular, for any row $j$ of sketch $\mathcal{S}$, component $i$ is hashed into bin $\mathcal{S}(j, h_j(i))$. In addition to $h_j$, Count-Sketch uses $d$ random sign functions to map the components of the vectors randomly to $\{+1,\ -1\}$. i.e., $s_i:\{1,\ 2,..., \ D\} \rightarrow \{+1,-1\}$ A picture of this sketch data structure with three hash functions in shown inside Fig.~\ref{fig:mission_schem}. 

The Count-Sketch supports two operations: UPDATE(item i, increment $\Delta$) and QUERY(item i). The UPDATE operation updates the sketch with any observed increment. More formally, for an increment $\Delta$ to an item $i$, the sketch is updated by adding $s_j(i)\Delta$ to the cell $\mathcal{S}(j,h_j(i))$ $\forall  j \in \{1,\ 2,...,\ d\}$. The QUERY operation returns an estimate for component $i$, the median of all the $d$ different associated counters.

It has been shown that, for any sequence of streaming updates (addition or subtraction) to the vector ${\boldsymbol \beta}$, Count-Sketch provides an unbiased estimate of any component $i$, $\widehat{{\boldsymbol \beta}_i}$ such that the following holds with high probability,
\begin{align}
\label{eq:CS_accu}
 {\boldsymbol \beta}_i - \epsilon ||{\boldsymbol \beta}||_2\ \le \ \widehat{{\boldsymbol \beta}_i}   \ \le \ {\boldsymbol \beta}_i + \epsilon ||{\boldsymbol \beta}||_2.
\end{align}
It can be shown that the Eq.~(\ref{eq:CS_accu}) is sufficient to achieve near-optimal guarantees for sparse recovery with the given space budget. Furthermore, these guarantees also meet the best compressed sensing lower bounds in terms of the number of counters (or measurements) needed for sparse recovery~\cite{indyk2013sketching}.

\section{Problem Formulation}
Consider the feature selection problem in the ultra high-dimensional setting: We are given the dataset $({\bf X}_i,y_i)$ for $i\in[n]=\{1,2,\dots,n\}$, where ${\bf X}_i\in \mathbb{R}^p$ and $y_i \in \mathbb{R}$ denote the $i^\text{th}$ measured and response variables. We are interested in finding the $k$-sparse ($k$ non-zero entries) feature vector (or regressor)  ${\boldsymbol \beta} \in \mathbb{R}^p$ from the optimization problem
\begin{align}
\label{eq:main_equation}
\min_{\|\boldsymbol\beta\|_0=k}{\| {\bf y}-{\bf X}{\boldsymbol\beta}\|_2},
\end{align}
where ${\bf X}=\{{\bf X}_1,{\bf X}_2,\dots,{\bf X}_n\}$ and ${\bf y}=[y_1,y_1,\dots,y_n]$ denote the data matrix and response vector and the $\ell_0$-norm $\|\boldsymbol\beta\|_0$ counts the number of non-zero entries in $\boldsymbol\beta$. 

We are interested in solving the feature selection problem for ultra high-dimensional datasets where the number of features $p$ is so large that a dense vector (or matrix) of size $p$ cannot be stored explicitly in memory. 

\subsection{Hard Thresholding Algorithms}
\label{sec:IHT}

Among the menagerie of feature selection algorithms, the class of hard thresholding algorithms have the smallest memory footprint: Hard thresholding algorithms retain only the top-$k$ values and indices of the entire feature vector using $\mathcal{O}(k \text{log}(p))$ memory \cite{jain2014iterative, blumensath2009iterative}. The \emph{iterative hard thresholding} (IHT) algorithm generates the following iterates for the $i^\text{th}$ variable in an stochastic gradient descent (SGD) framework
\begin{align}
{\boldsymbol \beta}^{t+1} \leftarrow  H_k( {\boldsymbol \beta}^{t} -2 \lambda \left( y_i - {\bf X}_i {\boldsymbol \beta}^t \right)^T {\bf X}_i)
\end{align}

The sparsity of the feature vector $\boldsymbol \beta^t$, enforced by the hard thresholding operator $H_k$, alleviates the need to store a vector of size $\mathcal{O}(p)$ in the memory in order to keep track of the changes of the features over the iterates.

Unfortunately, because it only retains the top-$k$ elements of ${\boldsymbol \beta}$, the hard thresholding procedure greedily discards the information of the non top-$k$ coordinates from the previous iteration. In particular, it clips off coordinates that might add to the support set in later iterations. This drastically affects the performance of hard thresholding algorithms, especially in real-world scenarios where the design matrix ${\bf X}$ is not random, normalized, or well-conditioned. In this regime, the gradient terms corresponding to the true support typically arrive in lagging order and are prematurely clipped in early iterations by $H_k$. The effect of these \emph{lagging gradients} is present even in the SGD framework, because the gradients are quite noisy, and only a small fraction of the energy of the true gradient is expressed in each iteration. It is not difficult to see that these small energy, high noise signals can easily cause the greedy hard thresholding operator to make sub-optimal or incorrect decisions. Ideally, we want to accumulate the gradients to get enough confidence in signal and to average out any noise. However, accumulating gradients will make the gradient vector dense, blowing up the memory requirements. This aforementioned problem is in fact symptomatic of all other thresholding variants including the \emph{Iterative algorithm with inversion} (ITI) \cite{maleki2009coherence} and the \emph{Partial hard thresholding} (PHT) algorithm \cite{jain2017partial}.

\section{The \mission{} Algorithm}
\label{sec:algorithm}

We now describe the \mission{} algorithm. First, we initialize the Count-Sketch $\mathcal{S}$ and the feature vector ${\boldsymbol \beta}^{t=0}$ with zeros entries. The Count-Sketch hashes a $p$-dimensional vector into $\mathcal{O}(\text{log}^2 p)$ buckets (Recall Fig.~\ref{fig:mission_schem}). We discuss this particular choice for the size of the Count-Sketch and the memory-accuracy trade offs of \mission{} in Sections $\ref{sec:mem_log_scale}$ and $\ref{sec:feature_extraction}$. 

At iteration $t$, \mission{} selects a random row ${\bf X}_i$ from the data matrix ${\bf X}$ and computes the stochastic gradient update term using the learning rate $\lambda$. $g_i~=~2 \lambda \left( y_i - {\bf X}_i {\boldsymbol \beta}^t \right)^T {\bf X}_i $ i.e. the usual gradient update that minimizes the \emph{unconstrained} quadratic loss $\|{\bf y} - {\bf X}{\boldsymbol \beta}\|_2^2$. The data vector ${\bf X}_i$ and the corresponding stochastic gradient term are sparse. We then add the non-zero entries of the stochastic gradient term $\{g_{ij}: \forall{j}\,\,\, g_{ij}>0\}$ to the Count-Sketch $\mathcal{S}$. Next, \mission{} queries the top-$k$ values of the sketch to form ${\boldsymbol \beta}^{t+1}$. We repeat the same procedure until convergence. \mission{} returns the top-$k$ values of the Count-Sketch as the final output of the algorithm. The \mission{} algorithm is detailed in Alg.~$\ref{alg:algorithm}$. \mission{} easily extends to other loss functions such as the hinge loss and logistic loss.

\begin{algorithm}[tb]
   \caption{\mission{}}
   \label{alg:algorithm}
\begin{algorithmic}
   \STATE {\bf Initialize}: $\beta^0 = 0$, $\mathcal{S}$ (Count-Sketch), $\lambda$ (Learning Rate)
   \WHILE {not stopping criteria}
   \STATE Find the gradient update $g_i=\lambda \left( 2\left( y_i - {\bf X}_i {\boldsymbol \beta}^t \right)^T {\bf X}_i \right)$ 
   \STATE Add the gradient update to the sketch $g_i \rightarrow \mathcal{S}$
   \STATE Get the top-$k$ heavy-hitters from the sketch ${\boldsymbol \beta}^{t+1} \leftarrow  \mathcal{S}$
   \ENDWHILE
   \STATE {\bf Return:} The top-$k$ heavy-hitters from the Count-Sketch
\end{algorithmic}
\end{algorithm}

{\bf \mission{} is Different from Greedy Thresholding:} Denote the gradient vector update at any iteration $t$ as $u_t$. It is not difficult to see that starting with an all-zero vector ${\boldsymbol \beta}_0$, at any point of time $t$, the Count-Sketch state is equivalent to the sketch of the vector $\sum_{i=1}^t u_t$. In other words, the sketch aggregates the compressed aggregated vector. Thus, even if an individual SGD update is noisy and contains small signal energy, thresholding the Count-Sketch is based on the average update over time. This averaging produces a robust signal that cancels out the noise. We can therefore expect \mission{} to be superior over thresholding. 

In the supplementary materials, we present initial theoretical results on the convergence of \mission{}. Our results show that, under certain assumptions, the full-gradient-descent version of \mission{} converges \emph{geometrically} to the true parameter ${\boldsymbol \beta}\in \mathbb{R}^p$ up to some additive constants. The exploration of these assumptions and the extension to the SGD version of \mission{} are exciting avenues for future work. 

{\bf Feature Selection with the Ease of Feature Hashing:} As argued earlier, the features are usually represented with strings, and we do not have the capability to map each string to a unique index in a vector without spending $O(p)$ memory. Feature hashing is convenient, because we can directly access every feature using hashes. We can use any lossy hash function for strings. \mission{} only needs a few independent hash functions (3 in our Count-Sketch implementation) to access any component. The top-$k$ estimation is done efficiently using a heap data structure of size $k$. Overall, we only access the data using efficient hash functions, which can be easily implemented in large-scale systems. 

\section{Simulations}
\label{sec:simulation}
We designed a set of simulations to evaluate \mission{} in a controlled setting. In contrast to the ultra large-scale, real-world experiments of Section~\ref{sec:experiments}, in the section the data matrices are drawn from a random Gaussian distribution and the ground truth features are known.

\begin{table*}[t]
\small
\centering
\caption{Comparison of \mission{} against hard thresholding algorithms in subset selection under adversarial effects. We first report the percentage of instances in which the algorithms accurately find the solution (ACC) with no attenuation ($\alpha = 1$) over 100 random trials. We then report the mean of the maximum level of attenuation $\alpha$ applied to the columns of design ${\bf X}$ before the algorithms fail to recover the support of ${\boldsymbol \beta}$ (over the trials that all algorithms can find the solution with $\alpha=1$).}
\vspace{0.2cm}
\scalebox{1}{%
\begin{tabular}{lccccccccccccc}
\hline
($n$, $k$) &&\multicolumn{3}{c}{\mission{}}  & \multicolumn{3}{c}{IHT }  &  \multicolumn{3}{c}{ITI}  & \multicolumn{2}{c}{PHT}   \\
\hline 
&& ACC$_{\alpha=1}$ & $\alpha$ && ACC$_{\alpha=1}$ & $\alpha$ && ACC$_{\alpha=1}$& $\alpha$ && ACC$_{\alpha=1}$ & $\alpha$  \\
\hline
(100, 2)    &&  {\bf 100}\%  & {\bf 2.68 $\pm$ 0.37} &&  {\bf 100}\%  & 1.49 $\pm$ 0.33 && 91\%  & 1.33 $\pm$ 0.23 && 64\%  & 2.42 $\pm$ 0.87  \\ 
(100, 3)    &&  {\bf 100}\%  & {\bf 2.52 $\pm$ 0.36} &&  92\%            & 1.36 $\pm$ 0.46 && 70\%  & 1.15 $\pm$ 0.20 && 42\% & 2.05 $\pm$ 0.93 \\ 
(100, 4)    &&  {\bf 100}\%  & {\bf 2.53 $\pm$ 0.23} &&  72\%            & 1.92 $\pm$ 0.91 && 37\%  & 1.03 $\pm$ 0.09 && 39\% & 2.13 $\pm$ 1.07 \\ 
(200, 5)    &&  {\bf 100}\%  & {\bf 4.07 $\pm$ 0.36} &&  99\%            &  2.34 $\pm$ 1.12 && 37\%  & 1.15 $\pm$ 0.22 && 83\% & 2.75 $\pm$ 1.30 \\ 
(200, 6)    &&  {\bf 100}\%  & {\bf 4.17 $\pm$ 0.24} &&  97\%            & 2.64 $\pm$ 1.14 && 23\%  & 1.11 $\pm$ 0.12 && 73\% & 2.26 $\pm$ 1.33 \\ 
(200, 7)    &&  {\bf 100}\%  & {\bf 4.07 $\pm$ 0.11} &&  83\%            & 1.64 $\pm$ 1.01 && 14\%  & 1.11 $\pm$ 0.12 && 75\% & 3.39 $\pm$ 1.36 \\ 

\hline
\end{tabular}}
\label{table:attenuation}
\end{table*}

\subsection{Phase Transition}
\label{sec:phase_trans}

We first demonstrate the advantage of \mission{} over greedy thresholding in feature selection. For this experiment, we modify \mission{} slightly to find the root of the algorithmic advantage of \mission{}: we replace the Count-Sketch with an ``identity'' sketch, or a sketch with a single hash function, $h(i) = i$. In doing so, we eliminate the complexity that Count-Sketch adds to the algorithm, so that the main difference between \mission{} and IHT is that \mission{} accumulates the gradients. To improve stability, we scale the non top-$k$ elements of $\mathcal{S}$ by a factor $\gamma \in (0, 1)$ that begins very near 1 and is gradually decreased until the algorithm converges. \emph{Note:} it is also possible to do this scaling in the Count-Sketch version of \mission{} efficiently by exploiting the linearity of the sketch. 

Fig.~\ref{fig:phase_transition} illustrates the empirical phase transition curves for sparse recovery using \mission{} and the hard thresholding algorithms. The phase transition curves show the points where the algorithm successfully recovers the features in $>50\%$ of the random trails. \mission{} shows a better phase transition curve compared to IHT by a considerable gap. 

\begin{figure}[t]
\vspace{-0.1in}
\centering
\includegraphics[width=0.33\textwidth]{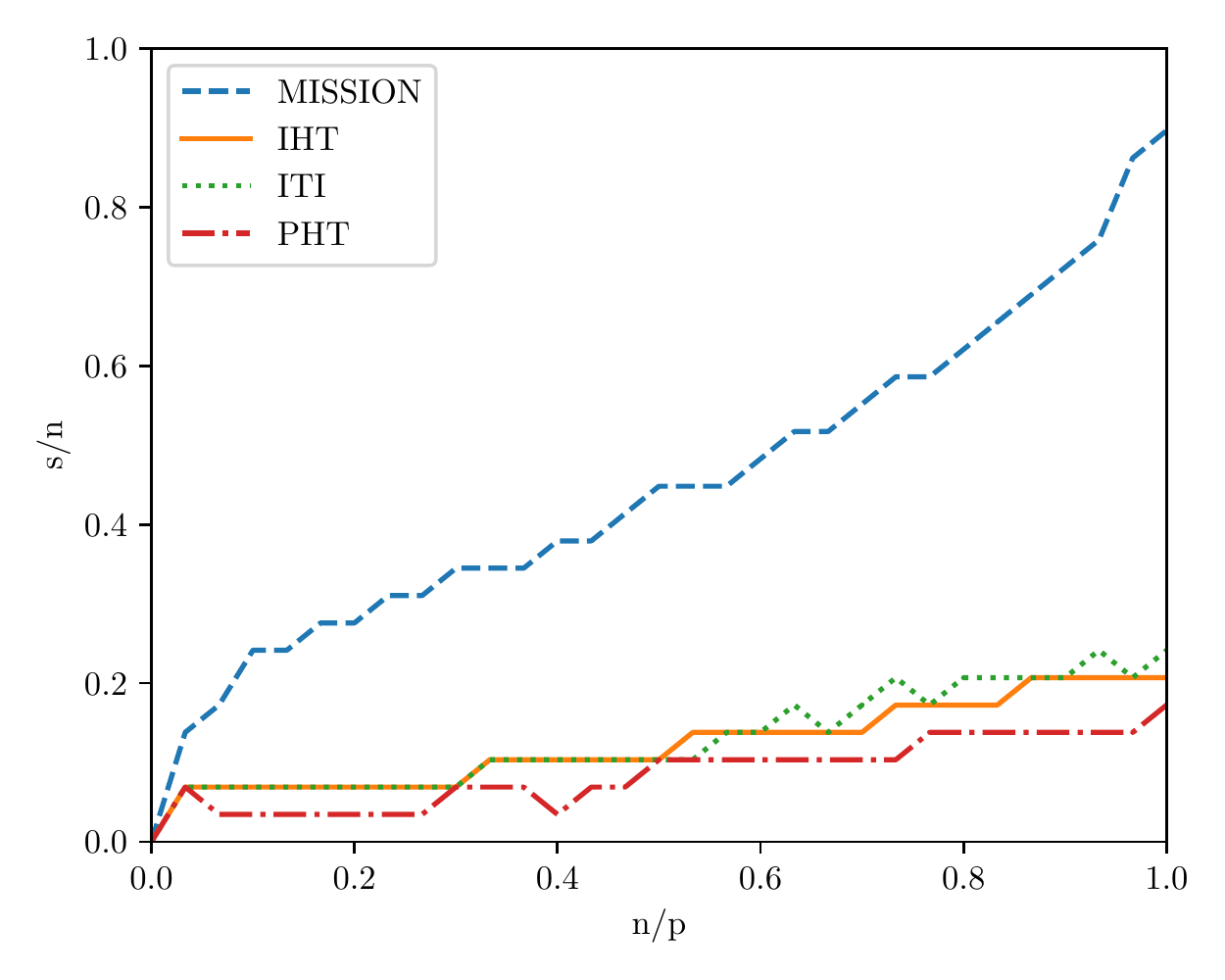}
\vspace{-0.2in}
\caption{Empirical phase transition in recovering a binary feature vector ${\boldsymbol \beta}$ in $p = 1000$-dimensional space with a Gaussian data matrix ${\bf X}$. We illustrate the empirical $50\%$ probability of success curves averaged over $T = 20$ trials. \mission{} outperforms the thresholding algorithms by a large margin.}
\label{fig:phase_transition}
\vspace{-0.2cm}
\end{figure}

\subsection{Lagging Gradient: Superiority of Count-Sketches over Greedy Thresholding}
\label{sec:lag_grad}

A major problem with the IHT algorithm, especially in large-scale SGD settings, is with thresholding the coordinates with small gradients in the earlier iterations. IHT misses these coordinates, since they become prominent only after the gradients accumulate with the progression of the algorithm. The problem is amplified with noisy gradient updates such as SGD, which is unavoidable for large datasets.

This phenomenon occurs frequently in sparse recovery problems. For example, when the coordinates that correspond to the columns of the data matrix with smaller energy lag in the iterations of gradient descent algorithm, IHT thresholds these \emph{lagging-gradient} coordinates in first few iterations, and they never show up again in the support. In contrast, \mission{} retains a footprint of the gradients of all the previous iterations in the Count-Sketch. When the total sum of the gradient of a coordinate becomes prominent, the coordinate joins the support after querying the top-$k$ heavy hitters from the Count-Sketch. We illustrate this phenomena in sparse recovery using synthetic experiments. We recover sparse vector ${\boldsymbol \beta}$ from its random linear measurements ${\bf y} = {\bf X}{\boldsymbol \beta}$, where the energy of ${\bf X}$ is imbalanced across its columns. In this case, the gradients corresponding to the columns (coordinates) with smaller energy typically lag and are thresholded by IHT.

To this end, we first construct a random Gaussian data matrix ${\bf X} \in \mathbb{R}^{900\times1000}$, pick a sparse vector ${\boldsymbol \beta}$ that is supported on an index set $\mathcal{I}$, and then attenuate the energy of the columns of $\bf X$ supported by the indices in $\mathcal{I}$ by an attenuation factor of $\alpha=\{1,1.25,1.5,1.75,2,\dots,5\}$. Note that $\alpha=1$ implies that no attenuation is applied to the matrix. In Table~\ref{table:attenuation}, we report the maximum attenuation level applied to a column of data matrix {\bf X} before the algorithms fail to fully recover the support set $\mathcal{I}$ from ${\bf y}= {\boldsymbol \beta} {\bf X}$. We observe that \mission{} is consistently and up to three times more robust against adversarial attenuation of the columns of the data matrix in various design settings.

The robustness of \mission{} to the attenuation of the columns of  {\bf X} in sparse recovery task suggests that the Count-Sketch data structure enables gradient-based optimization methods such as IHT to store a footprint (or sketch) of all the gradients from the previous iterations and deliver them back when they become prominent.

\subsection{Logarithmic Scaling of the Count-Sketch Memory in \mission{}}
\label{sec:mem_log_scale}
In this section we demonstrate that the memory requirements of \mission{} grows polylogarithmically in the dimension of the problem $p$. We conduct a feature selection experiment with a data matrix ${\bf X} \in \mathbb{R}^{100\times p}$ whose entries are drawn from i.i.d. random Gaussian distributions with zero mean and unit variance. We run \mission{} and IHT to recover the feature vector ${\boldsymbol \beta}$ from the output vector ${\bf y} = {\bf X}{\boldsymbol \beta}$, where the feature vector ${\boldsymbol \beta}$ is a $k=5$-sparse vector with random support. We repeat the  same experiment $1000$ times with different realizations for the sparse feature vector ${\boldsymbol \beta}$ and report the results in Fig.~$\ref{fig:MISSION_memory}$. The left plot illustrates the feature selection accuracy of the algorithms as the dimension of the problem $p$ grows. The right plot illustrates the minimum memory requirements of the algorithms to recover the features with $100\%$ accuracy. 

The plots reveal an interesting phenomenon. The size of the Count-Sketch in \mission{} scales only polylogarithmically with the dimension of the problem. This is surprising since the aggregate gradient in a classical SGD framework becomes typically dense in early iterations and thus requires a memory of order $\mathcal{O}(p)$. \mission{}, however, stores only the \emph{essential} information of the features in the sketch using a poly-logarithmic sketch size. Note that IHT sacrifices accuracy to achieve a small memory footprint. At every iteration IHT eliminates all the information except for the top-$k$ features. We observe that, using only a logarithmic factor more memory, \mission{} has a significant advantage over IHT in recovering the ground truth features.

\begin{figure}
\begin{center}
\vspace*{-2.5mm}
     \includegraphics[width=0.45\textwidth]{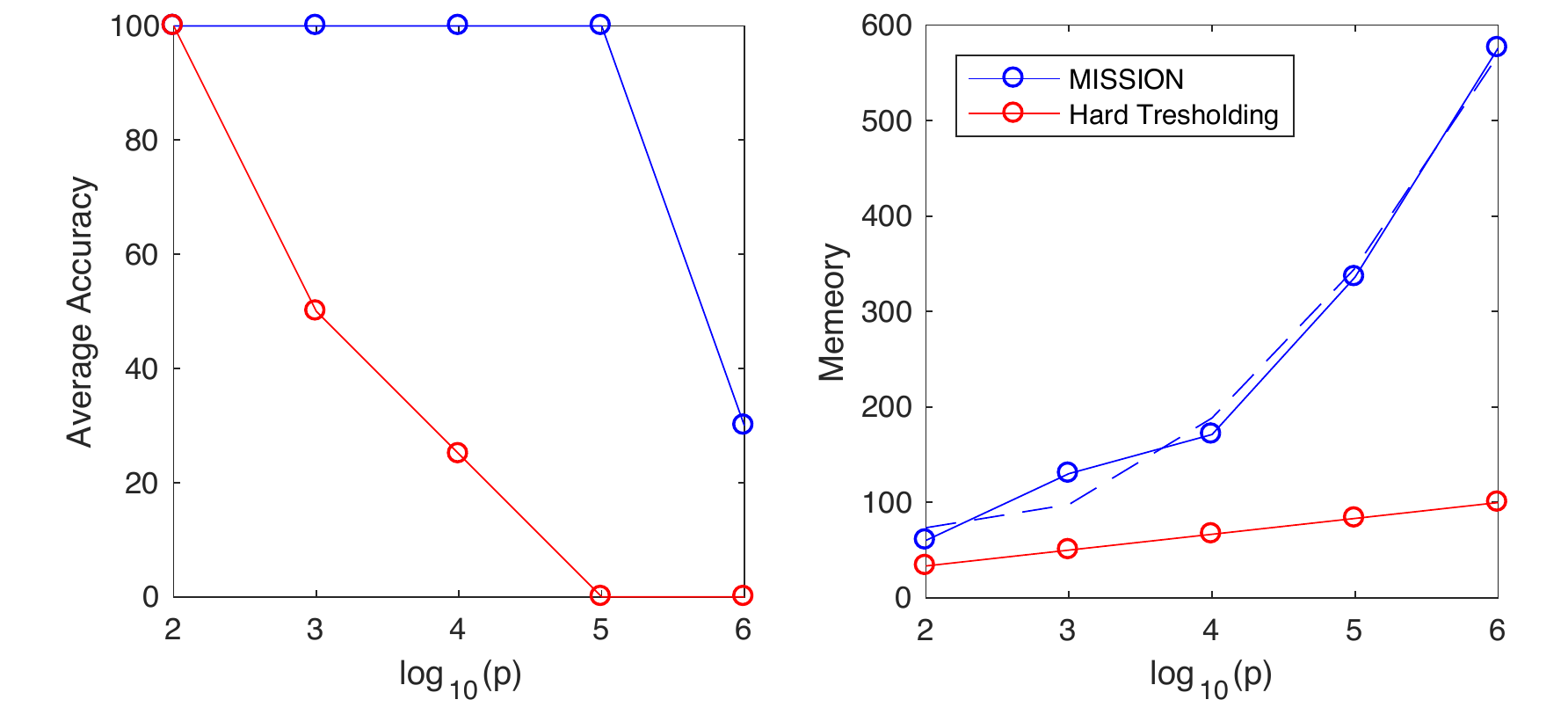}
        \caption{\small Feature selection accuracy and memory requirements of \mission{} and Hard Thresholding. The memory requirements of \mission{} grows polylogarithmcially $\sim \mathcal{O}(\text{log}^2(p))$ (dotted line illustrates quadratic fit) in $p$. With only a logarithmic factor more memory, \mission{} has significant advantage over Hard Thresholding in terms of feature selection accuracy.}\vspace*{-3.5mm}
        \label{fig:MISSION_memory}
        \end{center}
 \end{figure}

\section{Experiments}
\label{sec:experiments}
All experiments were performed on a single machine, 2x Intel Xeon E5-2660 v4 processors (28 cores / 56 threads) with 512 GB of memory. The code \footnote{\url{ https://github.com/rdspring1/MISSION}} for training and running our randomized-hashing approach is available online. We designed the experiments to answer these questions:
\begin{enumerate}
    \itemsep-0.25em
    \item Does \mission{} outperform IHT in terms of classification accuracy? In particular, how much does myopic thresholding affect IHT in practice?
    \item How well does \mission{} match the speed and accuracy of feature hashing (FH)?
    \item How does changing the number of top-$k$ features affect the accuracy and behaviour of the different methods?
    \item What is the effect of changing the memory size of the Count-Sketch data structure on the classification accuracy of \mission{} in read-world datasets?
    \item Does \mission{} scale well in comparison to the different methods on the ultra large-scale datasets ($>350$ GB in size)?
\end{enumerate}

\subsection{Large-scale Feature Extraction}
\label{sec:feature_extraction}

{\bf Datasets:} We used four datasets in the experiments: 1) KDD2012, 2) RCV1, 3) Webspam--Trigram, 4) DNA \footnote{\url{http://projects.cbio.mines-paristech.fr/largescalemetagenomics/}}. The statistics of these datasets are summarized in Table \ref{table:dataset}.

\begin{table} [ht]
\vspace{-0.45cm}
\caption{Feature extraction dataset statistics.}
\vspace{-0.25cm}
\begin{center}
    \begin{tabular}{ |l|l|l|l| } 
    \hline
    Dataset & Dim ($p$) & Train Size ($n$) & Test Size \\
    \hline
    KDD 2012 & 54,686,452 & 119,705,032 & 29,934,073 \\
    RCV1 & 47,236 & 20,242 & 677,399\\
    Webspam & 16,609,143 & 280,000 & 70,000 \\
    DNA (Tiny) & 14,890,408 & 1,590,000 & 207,468 \\
    \hline
    \end{tabular}
\end{center}
\vspace{-0.5cm}
\label{table:dataset}
\end{table}

The DNA metagenomics dataset is a multi-class classification task where the model must classify 15 different bacteria species using DNA $K$-mers. We sub-sampled the first 15 species from the original dataset containing 193 species. We use all of the species in the DNA Metagenomics dataset for the large-scale experiments (See Section \ref{sec:large_scale}). Following standard procedures, each bacterial species is associated with a reference genome. Fragments are sampled from the reference genome until each nucleotide is covered $c$ times on average. The fragments are then divided into $K$-mer sub-strings. We used fragments of length 200 and $K$-mers of length 12. Each model was trained and tested with mean coverage $c=\{0.1,1\}$ respectively. For more details, see \cite{vervier2016large}. The feature extraction task is to find the DNA $K$-mers that best represent each bacteria class.

We implemented the following approaches to compare and contrast against our approach: For all methods, we used the logistic loss for binary classification and the cross-entropy loss for multi-class classification.

{\bf MISSION:} As described in Section \ref{sec:algorithm}.\\
{\bf Iterative Hard Thresholding (IHT):} An algorithm where, after each gradient update, a hard threshold is applied to the features. Only the top-$k$ features are kept active, while the rest are set to zero. Since the features are strings or integers, we used a sorted heap to store and manipulate the top-$k$ elements. This was the only algorithm we could successfully run over the large datasets on our single machine.\\
{\bf Batch IHT:} A modification to IHT that uses mini-batches such that the gradient sparsity is the same as the number of elements in the count-sketch. We accumulate features and then sort and prune to find the top-$k$ features. This accumulate, sort, prune process is repeated several times during training. {\bf Note:} This setup requires significantly more memory than \mission{}, because it explicitly stores the feature strings. The memory cost of maintaining a set of string features can be orders of magnitude more than the flat array used by \mission{}. See Bloom Filters~\cite{broder2004network} and related literature. This setup is not scalable to large-scale datasets.\\
{\bf Feature Hashing (FH):} A standard machine learning algorithm for dimensionality reduction that reduces the memory cost associated with large datasets. FH is not a feature selection algorithm and cannot identify important features. \citep{agarwal2014reliable}

\begin{figure*} [ht]
\vspace{-0.5in}
\begin{center}
\mbox{
\hspace{-0.25in}
\includegraphics[width=0.30\textwidth]{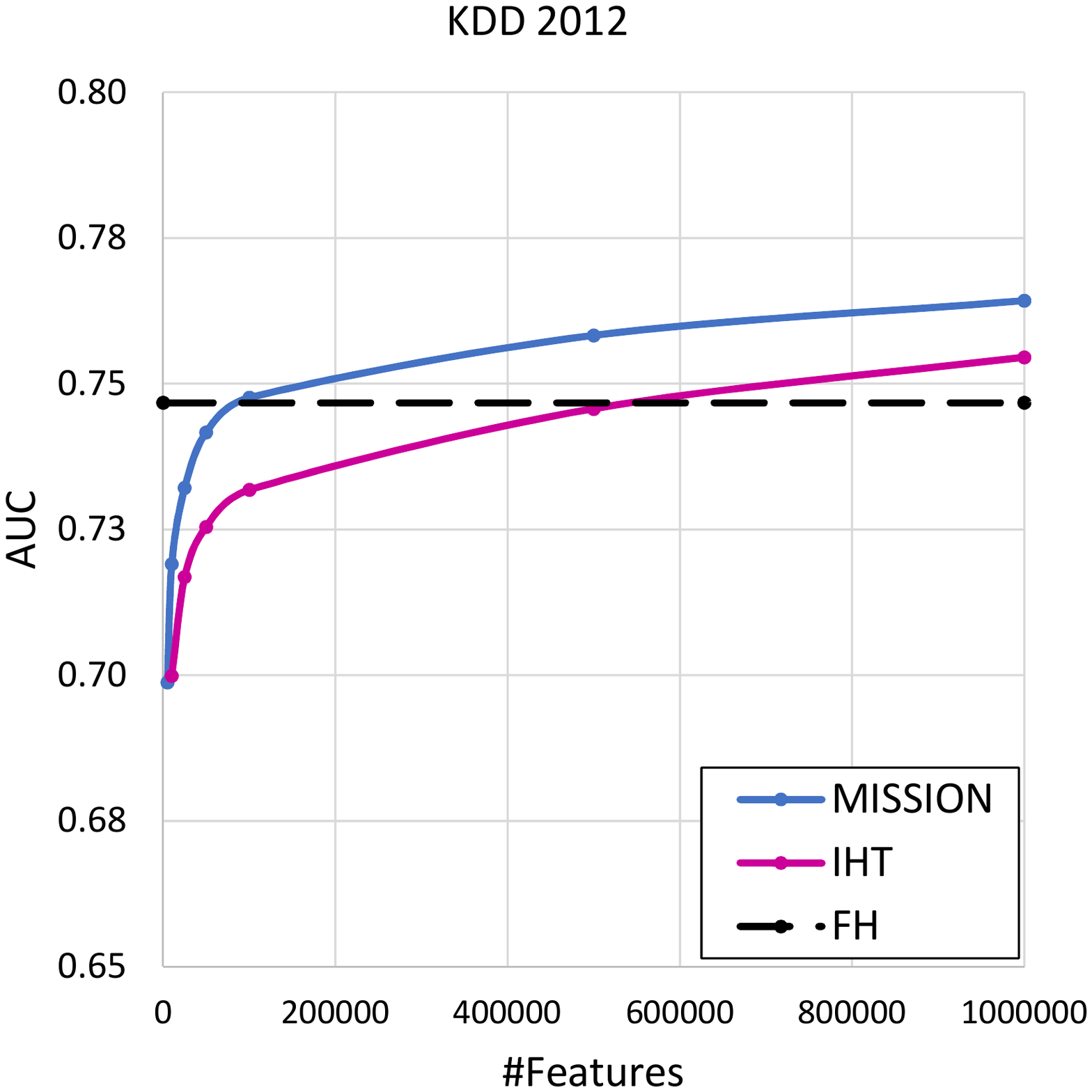}
\hspace{-0.4in}
\includegraphics[width=0.30\textwidth]{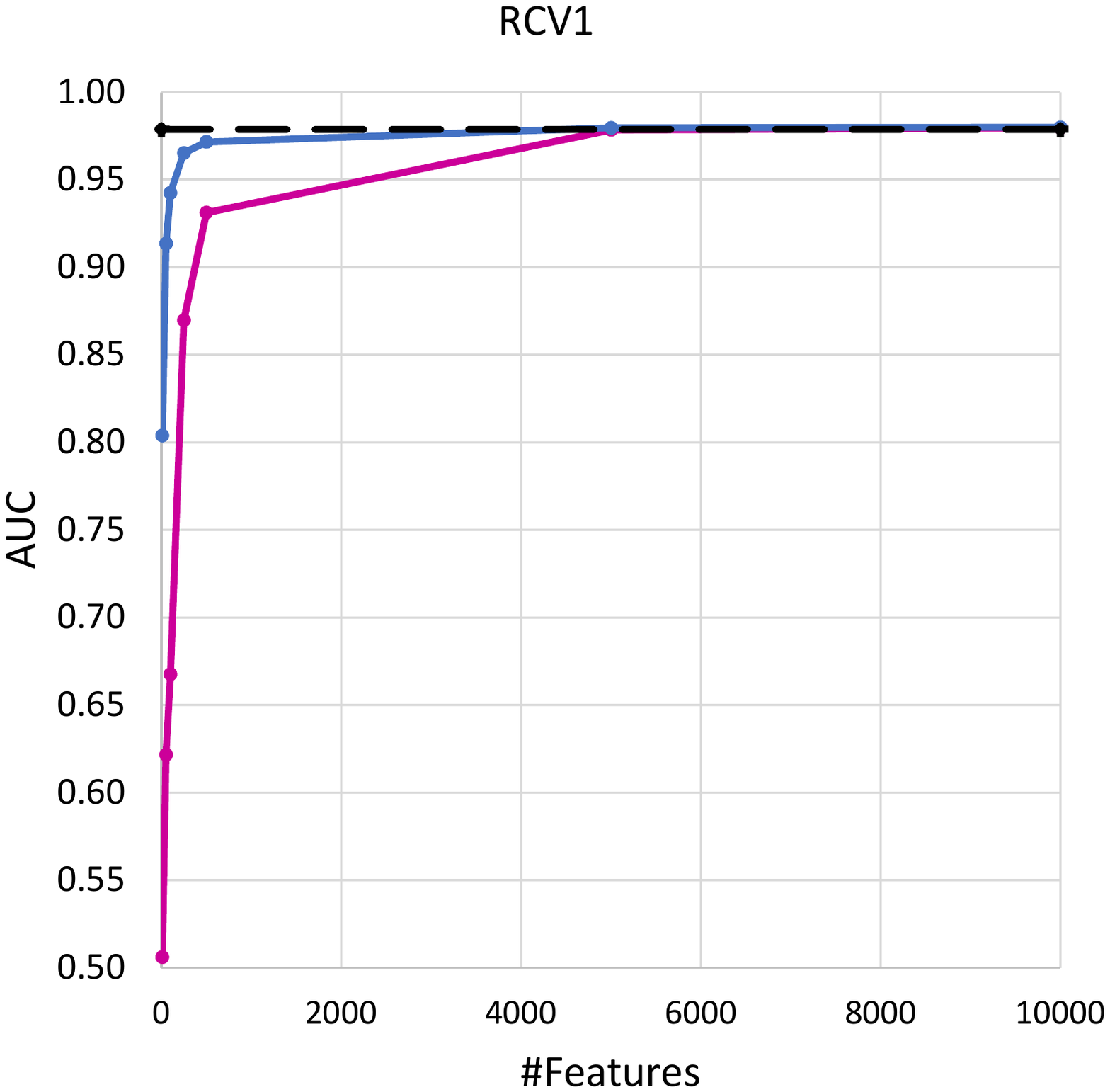}
\hspace{-0.4in}
\includegraphics[width=0.30\textwidth]{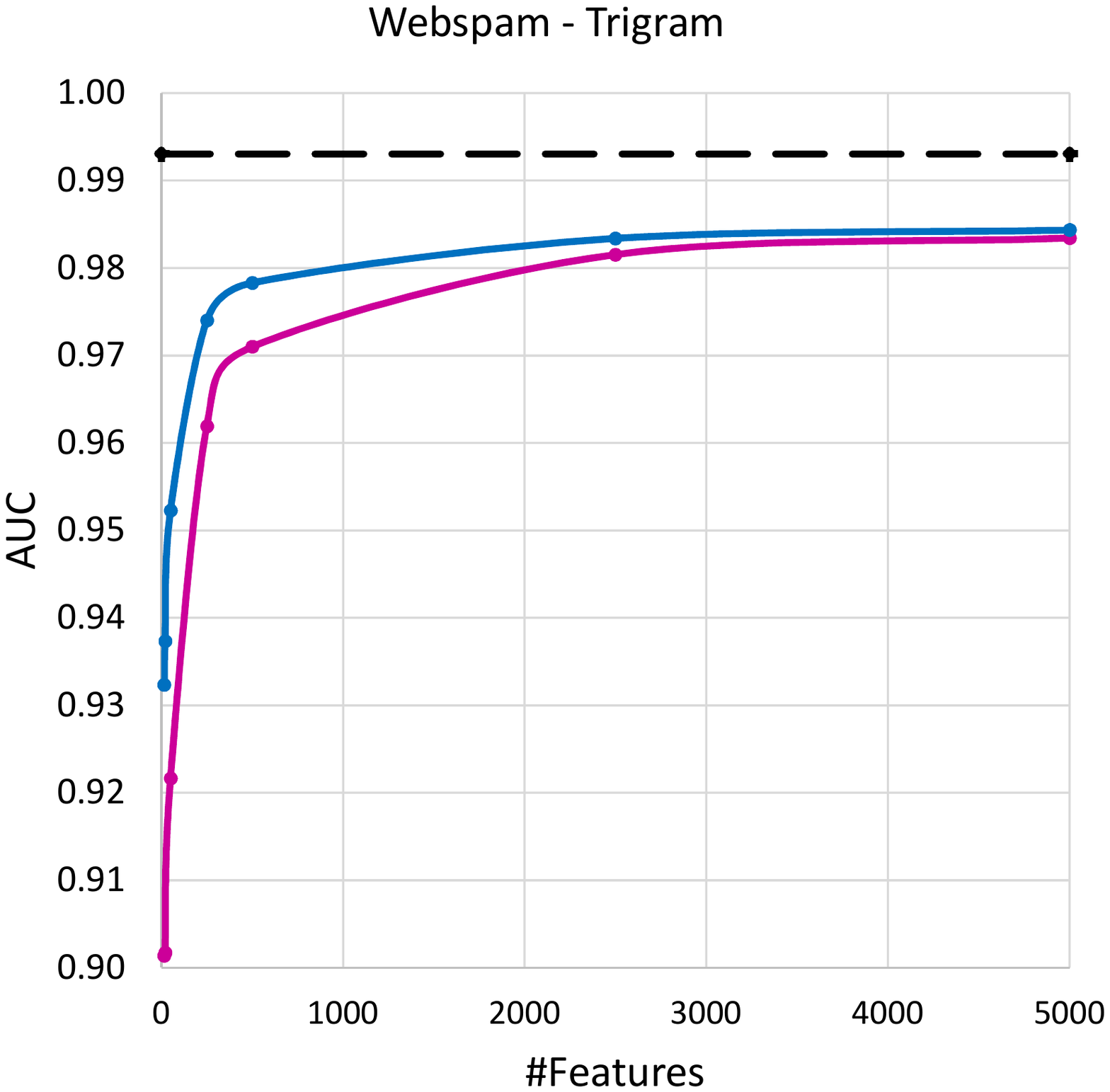}
\hspace{-0.4in}
\includegraphics[width=0.30\textwidth]{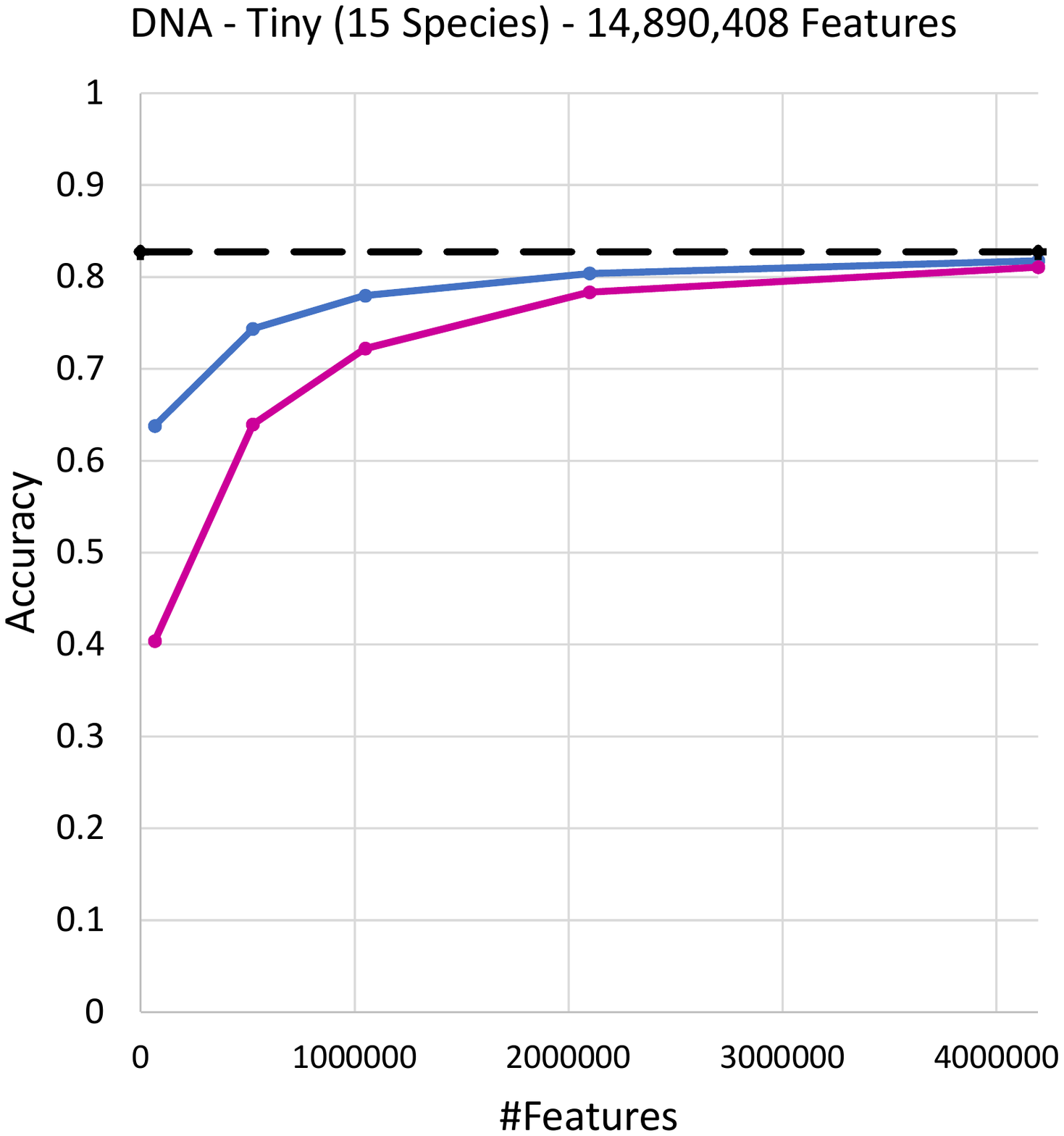}
\hspace{-0.25in}
}
\end{center}
\vspace{-0.7in}
\caption{Feature selection on the KDD-2012, RCV1, Webspam, and DNA Metagenomic (Tiny) datasets. FH is not a feature selection baseline. It is marked by a dashed black line, indicating that its performance is invariant to the number of top-$k$ features.}
\label{fig:feature_extract}
\end{figure*}

{\bf Experimental Settings:} The \mission{} and IHT algorithms searched for the same number of top-$k$ features. To ensure fair comparisons, the size of the Count-Sketch and the feature vector allocated for the FH model were equal. The size of the \mission{} and FH models were set to the nearest power of 2 greater than the number of features in the dataset. For all the experiments, the Count-Sketch data structure used 3 hash functions, and the model weights were divided equally among the hash arrays. For example, with the (Tiny) DNA metagenomics dataset, we allocated 24 bits or 16,777,216 weights for the FH model. Given 3 hash functions and 15 classes, roughly 372,827 elements were allocated for each class in the Count-Sketch.

{\bf \mission{}, IHT, FH Comparison:} Fig.~\ref{fig:feature_extract} shows that \mission{} surpasses IHT in classification accuracy in all four datasets, regardless of the number of features. In addition, \mission{} closely matches FH, which is significant because FH is allowed to model a much larger set of features than \mission{} or IHT. \mission{} is 2--4$\times$ slower than FH, which is expected given that \mission{} has the extra overhead of using a heap to track the top-$k$ features. 

\mission{}'s accuracy rapidly rises with respect to the number of top-$k$ features, while IHT's accuracy plateaus and then grows slowly to match \mission{}. This observation corroborates our insight that the greedy nature of IHT hurts performance. When the number of top-$k$ elements is small, the capacity of IHT is limited, so it picks the first set of features that provides good performance, ignoring the rest. On the other hand, \mission{} decouples the memory from the top-$k$ ranking, which is based on the aggregated gradients in the compressed sketch. By the linear property of the count-sketch, this ensures that the heavier entries occur in the top-$k$ features with high probability. 

{\bf Count-Sketch Memory Trade-Off:} Fig.~\ref{fig:DNA_Tradeoff} shows how \mission{}'s accuracy degrades gracefully, as the size of the Count-Sketch decreases. In this experiment, \mission{} only used the top 500K features for classifying the Tiny DNA metagenomics dataset. When the top-$k$ to Count-Sketch ratio is 1, then 500K weights were allocated for each class and hash array in the Count-Sketch data structure. The Batch IHT baseline was given 8,388,608 memory elements per class, enabling it to accumulate a significant number of features before thresholding to find the top-$k$ features. This experiment shows that \mission{} immediately outperforms IHT and Batch IHT, once the top-$k$ to Count-Sketch ratio is 1:1. Thus, MISSION provides a unique memory-accuracy knob at any given value of top-$k$. 

\begin{figure}[ht]
\vspace{-0.9in}
\begin{center}
\includegraphics[width=0.35\textwidth]{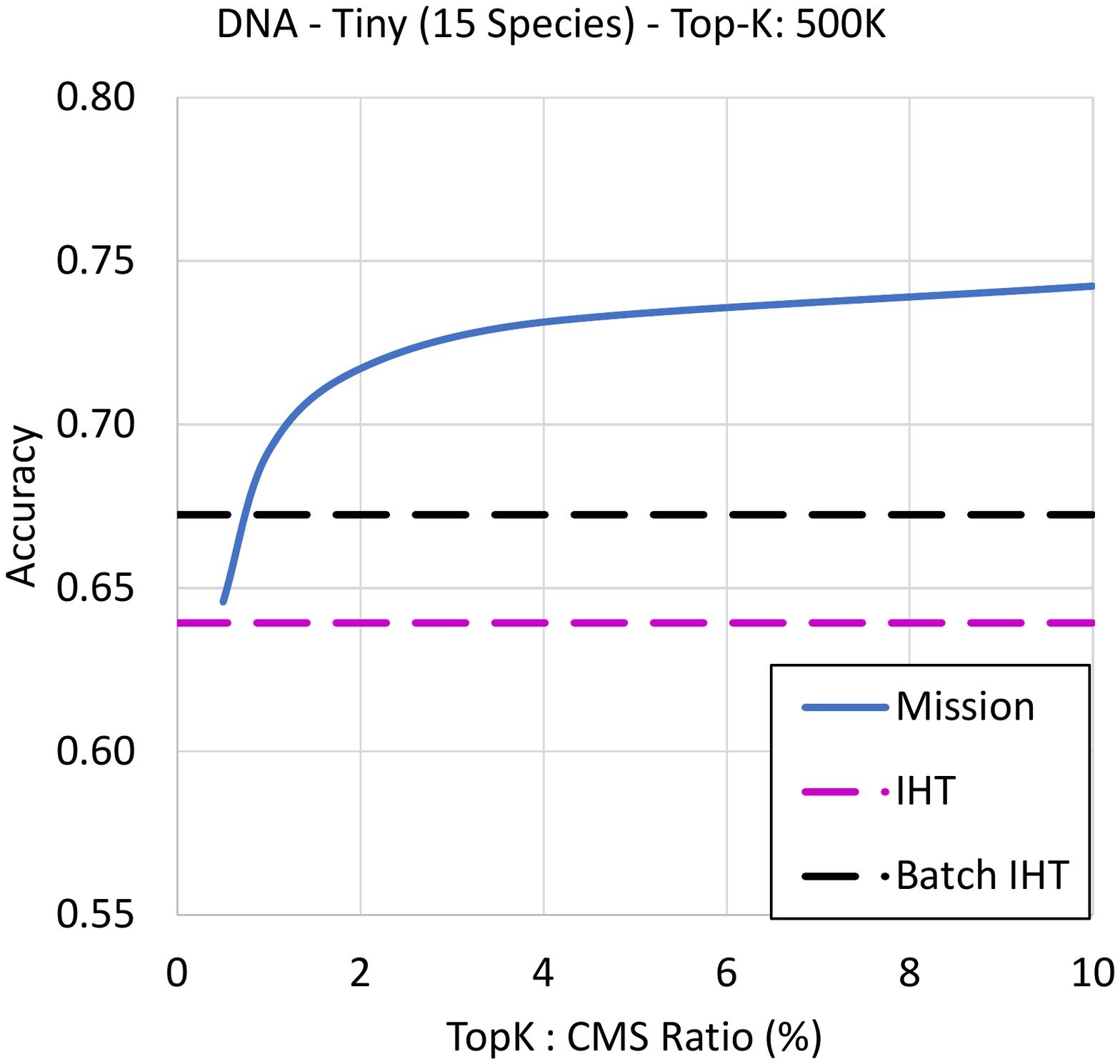}
\end{center}
\vspace{-0.9in}
\caption{Count-Sketch Memory/Accuracy Trade-off using the DNA Metagenomics (Tiny) dataset. The $x$-axis is the ratio of Count-Sketch memory per class and hash-array to the \# top-$k$ features per class. Memory Elements for Batch IHT: 8,388,608.}
\vspace{-0.5cm}
\label{fig:DNA_Tradeoff}
\end{figure} 

\subsection{Ultra Large-Scale Feature Selection}
\label{sec:large_scale}
Here we demonstrate that \mission{} can extract features from three large-scale datasets: Criteo 1TB, Splice-Site, and DNA Metagenomics.

\begin{table} [ht]
\vspace{-0.45cm}
\caption{Ultra Large-Scale dataset statistics}
\begin{center}
    \begin{tabular}{ |l|l|l|l| } 
    \hline
    Dataset & Dim ($p$) & Train Size ($n$) & Test Size \\
    \hline
    Criteo & 1M & 4,195,197,692 & 178,274,637 \\
    Splice-Site & 11.7M & 50,000,000 & 4,627,840 \\
    DNA & 17.3M & 13,792,260 & 354,285 \\
    \hline
    \end{tabular}
\end{center}
\vspace{-0.25cm}
\label{table:large_dataset}
\end{table}

{\bf Criteo 1TB:} The Criteo 1TB \footnote{\url{https://www.kaggle.com/c/criteo-display-ad-challenge}} dataset represents 24 days of click-through logs---23 days (training) + 1 day (testing). The task for this dataset is click-through rate (CTR) prediction---How likely is a user to click an ad? The dataset contains over 4 billion (training) and 175 million (testing) examples (2.5 TB of disk space). The performance metric is Area Under the ROC Curve (AUC). The VW baseline \footnote{\url{https://github.com/rambler-digital-solutions/criteo-1tb-benchmark}} achieved 0.7570 AUC score. \mission{} and IHT scored close to the VW baseline with 0.751 AUC using only the top 250K features.

\begin{table} [ht]
\vspace{-0.45cm}
\caption{Criteo 1TB. Top-K Features: 250K}
\begin{center}
    \begin{tabular}{ |c|c|c|c| } 
    \hline
    Metric & \mission{} & IHT & VW \\
    \hline
    AUC & 0.751 & 0.752 & 0.757 \\ 
    \hline
    \end{tabular}
\end{center}
\vspace{-0.25cm}
\label{fig:criteo1tb}
\end{table}

{\bf Splice-Site:}
The task for this dataset is to distinguish between true and fake splice sites using the local context around the splice site in-question. The dataset is highly skewed (few positive, many negative values), and so the performance metric is average precision (AP). Average precision is the precision score averaged over all recall scores ranging from 0 to 1. The dataset contains over 50 million (training) and 4.6 million (testing) examples (3.2 TB of disk space). All the methods were trained for a single epoch with a learning rate of 0.5. \mission{}, Batch IHT, and SGD IHT tracked the top 16,384 features. FH, \mission{}, and Batch IHT used 786,432 extra memory elements. \mission{} significantly outperforms Batch IHT and SGD IHT by 2.3\%. Also, unlike in Fig.~\ref{fig:DNA_Tradeoff}, the extra memory did not help Batch IHT, since it performed the same as SGD IHT. \mission{} (17.5 hours) is 15\% slower than FH (15 hours) in wall-clock running time.

\begin{table} [ht]
\vspace{-0.45cm}
\caption{Splice-Site: Top-$k$ features: 16,384. Memory elements: 786,432. \mission{} outperforms Batch IHT and SGD IHT. }
\begin{center}
    \begin{tabular}{ |c|c|c|c|c| } 
    \hline
    Metric & FH & \mission{} & Batch IHT & SGD IHT \\
    \hline
    AP & 0.522 & 0.510 & 0.498 & 0.498 \\
    \hline
    \end{tabular}
\end{center}
\vspace{-0.25cm}
\label{fig:splice_site}
\end{table}

{\bf DNA Metagenomics:} This experiment evaluates \mission{}'s performance on a medium-sized metagenomics dataset. The parameters from the Tiny (15 species) dataset in Section~\ref{sec:feature_extraction} are shared with this experiment, except the number of species is increased to 193. The size of a sample batch with mean coverage $c=1$ increased from 7 GB (Tiny) to 68 GB (Medium). Each round (mean coverage $c=0.25$) contains 3.45 million examples and about 16.93 million unique non-zero features (p). \mission{} and IHT tracked the top 2.5 million features per class. The FH baseline used $2^{31}$ weights, about 11.1 million weights per class, and we allocated the same amount of space for the Count-Sketch. Each model was trained on a dataset with coverage $c=5$.

Fig.~\ref{fig:DNA_Medium} shows the evolution of classification accuracy over time for \mission{}, IHT, and the FH baseline. After 5 epochs, \mission{} closely matches the FH baseline. {\bf Note:} \emph{\mission{} converges faster than IHT} such that \mission{} is 1--4 rounds ahead of IHT, with the gap gradually increasing over time. On average, the running time of \mission{} is 1--2$\times$ slower than IHT. However, this experiment demonstrates that since \mission{} converges faster, it actually needs less time to reach a certain accuracy level. Therefore, \mission{} is effectively faster and more accurate than IHT.

\begin{table} [ht]
\vspace{-0.9in}
\begin{center}
\includegraphics[width=0.35\textwidth]{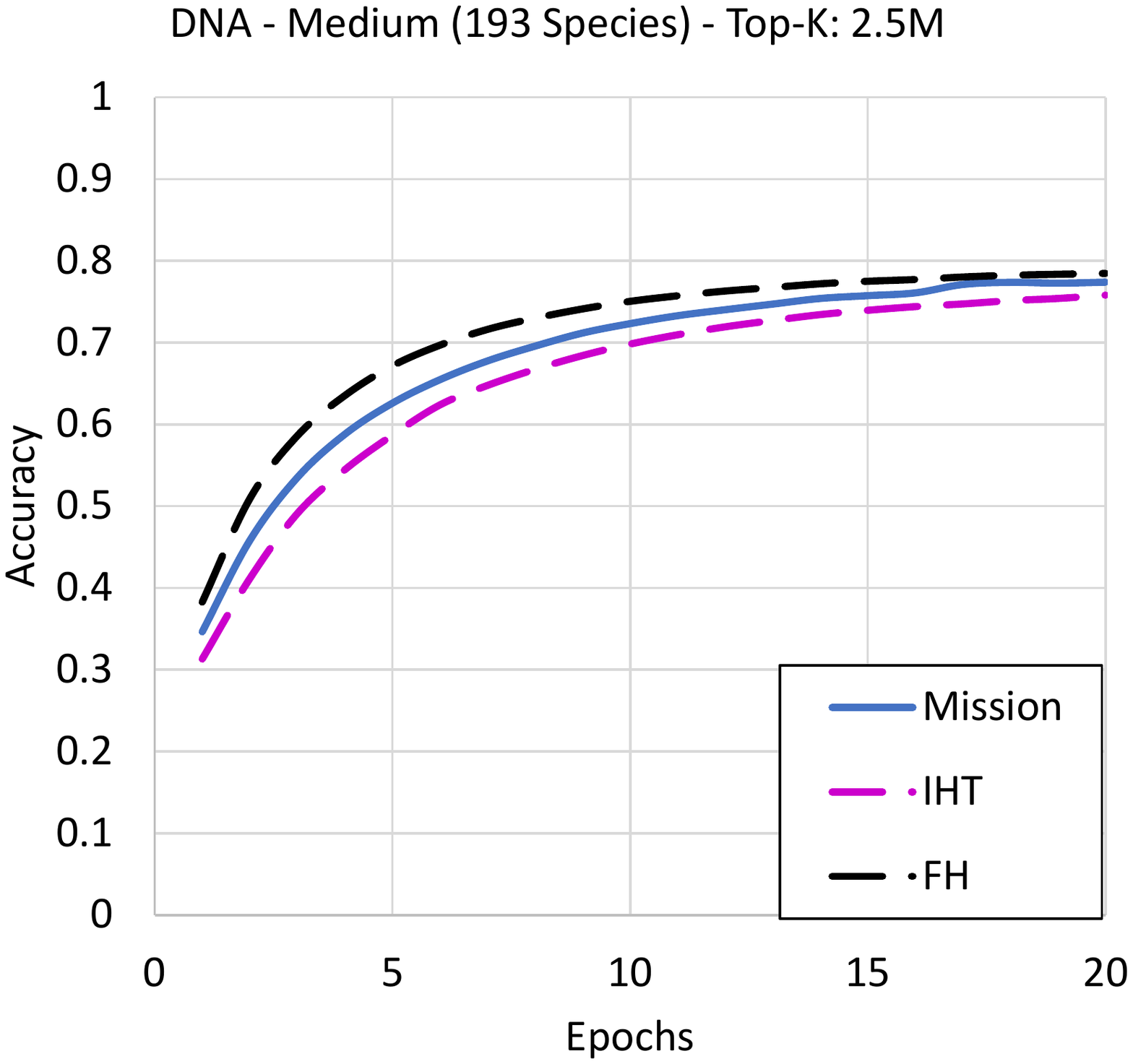}
\end{center}
\vspace{-0.9in}
\caption{Ultra Large-Scale Feature Selection for the DNA Metagenomics (Medium) Dataset (193 species), Mean Coverage $c=5$.}
\vspace{-0.5cm}
\label{fig:DNA_Medium}
\end{table}

\section{Implementation Details and Discussion}
\label{sec:discussion}
{\bf Scalability and Parallelism}:
IHT finds the top-$k$ features after each gradient update, which requires sorting the features based on their weights before thresholding. The speed of the sorting process is improved by using a heap data structure, but it is still costly per update. \mission{} also uses a heap to store its top-$k$ elements, but it achieves the same accuracy as IHT with far fewer top-$k$ elements because of the Count-Sketch. (Recall Section~\ref{sec:algorithm})

Another suggested improvement for the top-$k$ heap is to use lazy updates. Updating the weight of a feature does not change its position in the heap very often, but still requires an $O(\log n)$ operation. With lazy updates, the heap is updated only if it the change is significant. $|x_{t} - x_0| \ge \epsilon$, i.e. the new weight at time $t$ exceeds the original value by some threshold. This tweak significantly reduces the number of heap updates at the cost of slightly distorting the heap.

\section{Conclusion and Future Work}
In this paper, we presented \mission{}, a new framework for ultra large-scale feature selection that performs hard thresholding and SGD while maintaining an efficient, approximate representation for all features using a Count-Sketch data structure. \mission{} retains the simplicity of feature hashing without sacrificing the interpretability of the features.

Interaction features are important for scientific discovery with DNA Metagenomics \cite{Basu201711236}. Traditionally, the polynomial kernel trick enabled machine learning algorithms to explore this feature space implicitly without the exponential increase in dimensionality. However, this exponential cost is unavoidable with feature extraction. Going forward, we are interested in leveraging our \mission{} framework to explore pairwise or higher interaction features. 

\subsection*{Acknowledgements}

AAA, DL, GD, and RB were supported by the DOD Vannevar Bush Faculty Fellowship grant N00014-18-1-2047, NSF grant CCF-1527501, ARO grant W911NF-15-1-0316, AFOSR grant FA9550-14-1-0088, ONR grant N00014-17-1-2551, DARPA REVEAL grant HR0011-16-C-0028, and an ONR BRC grant for Randomized Numerical Linear Algebra. RS and AS were supported by NSF-1652131, AFOSR-YIP  FA9550-18-1-0152, and ONR BRC grant for Randomized Numerical Linear Algebra. The authors would also like to thank NVIDIA and Amazon for gifting computing resources. 

\balance
\bibliography{main_camera_ready.bib}

\begin{thebibliography}{20}
\providecommand{\natexlab}[1]{#1}
\providecommand{\url}[1]{\texttt{#1}}
\expandafter\ifx\csname urlstyle\endcsname\relax
  \providecommand{\doi}[1]{doi: #1}\else
  \providecommand{\doi}{doi: \begingroup \urlstyle{rm}\Url}\fi

\bibitem[Agarwal et~al.(2014)Agarwal, Chapelle, Dud{\'\i}k, and
  Langford]{agarwal2014reliable}
Agarwal, A., Chapelle, O., Dud{\'\i}k, M., and Langford, J.
\newblock A reliable effective terascale linear learning system.
\newblock \emph{Journal of Machine Learning Research}, 15\penalty0
  (1):\penalty0 1111--1133, 2014.

\bibitem[Aghazadeh et~al.(2016)Aghazadeh, Lin, Sheikh, Chen, Atkins, Johnson,
  Petrosino, Drezek, and Baraniuk]{aghazadeh2016universal}
Aghazadeh, A., Lin, A.~Y., Sheikh, M.~A., Chen, A.~L., Atkins, L.~M., Johnson,
  C.~L., Petrosino, J.~F., Drezek, R.~A., and Baraniuk, R.~G.
\newblock Universal microbial diagnostics using random dna probes.
\newblock \emph{Science advances}, 2\penalty0 (9):\penalty0 e1600025, 2016.

\bibitem[Basu et~al.(2018)Basu, Kumbier, Brown, and Yu]{Basu201711236}
Basu, S., Kumbier, K., Brown, J.~B., and Yu, B.
\newblock Iterative random forests to discover predictive and stable high-order
  interactions.
\newblock \emph{Proceedings of the National Academy of Sciences}, 2018.
\newblock ISSN 0027-8424.
\newblock \doi{10.1073/pnas.1711236115}.
\newblock URL \url{http://www.pnas.org/content/early/2018/01/17/1711236115}.

\bibitem[Blumensath \& Davies(2009)Blumensath and
  Davies]{blumensath2009iterative}
Blumensath, T. and Davies, M.~E.
\newblock Iterative hard thresholding for compressed sensing.
\newblock \emph{Applied and Computational Harmonic Analysis}, 27\penalty0
  (3):\penalty0 265--274, 2009.

\bibitem[Bray et~al.(2015)Bray, Pimentel, Melsted, and Pachter]{bray2015near}
Bray, N., Pimentel, H., Melsted, P., and Pachter, L.
\newblock Near-optimal {RNA}-{S}eq quantification.
\newblock \emph{arXiv preprint arXiv:1505.02710}, 2015.

\bibitem[Broder \& Mitzenmacher(2004)Broder and
  Mitzenmacher]{broder2004network}
Broder, A. and Mitzenmacher, M.
\newblock Network applications of bloom filters: A survey.
\newblock \emph{Internet mathematics}, 1\penalty0 (4):\penalty0 485--509, 2004.

\bibitem[Charikar et~al.(2002)Charikar, Chen, and
  Farach-Colton]{charikar2002finding}
Charikar, M., Chen, K., and Farach-Colton, M.
\newblock Finding frequent items in data streams.
\newblock In \emph{International Colloquium on Automata, Languages, and
  Programming}, pp.\  693--703. Springer, 2002.

\bibitem[Duchi et~al.(2008)Duchi, Shalev-Shwartz, Singer, and
  Chandra]{duchi2008efficient}
Duchi, J., Shalev-Shwartz, S., Singer, Y., and Chandra, T.
\newblock Efficient projections onto the {$\ell_1$}-ball for learning in high
  dimensions.
\newblock In \emph{Proceedings of the 25th International Conference on Machine
  Learning}, pp.\  272--279. ACM, 2008.

\bibitem[Indyk(2013)]{indyk2013sketching}
Indyk, P.
\newblock Sketching via hashing: {F}rom heavy hitters to compressed sensing to
  sparse fourier transform.
\newblock In \emph{Proceedings of the 32nd ACM SIGMOD-SIGACT-SIGAI Symposium on
  Principles of Database Systems}, pp.\  87--90. ACM, 2013.

\bibitem[Jain et~al.(2014)Jain, Tewari, and Kar]{jain2014iterative}
Jain, P., Tewari, A., and Kar, P.
\newblock On iterative hard thresholding methods for high-dimensional
  m-estimation.
\newblock In \emph{Advances in Neural Information Processing Systems}, pp.\
  685--693, 2014.

\bibitem[Jain et~al.(2017)Jain, Tewari, and Dhillon]{jain2017partial}
Jain, P., Tewari, A., and Dhillon, I.~S.
\newblock Partial hard thresholding.
\newblock \emph{IEEE Transactions on Information Theory}, 63\penalty0
  (5):\penalty0 3029--3038, 2017.

\bibitem[Langford et~al.(2009)Langford, Li, and Zhang]{langford2009sparse}
Langford, J., Li, L., and Zhang, T.
\newblock Sparse online learning via truncated gradient.
\newblock \emph{Journal of Machine Learning Research}, 10\penalty0
  (Mar):\penalty0 777--801, 2009.

\bibitem[Maleki(2009)]{maleki2009coherence}
Maleki, A.
\newblock Coherence analysis of iterative thresholding algorithms.
\newblock In \emph{Communication, Control, and Computing, 2009. Allerton 2009.
  47th Annual Allerton Conference on}, pp.\  236--243. IEEE, 2009.

\bibitem[McMahan et~al.(2013)McMahan, Holt, Sculley, Young, Ebner, Grady, Nie,
  Phillips, Davydov, Golovin, et~al.]{mcmahan2013ad}
McMahan, H.~B., Holt, G., Sculley, D., Young, M., Ebner, D., Grady, J., Nie,
  L., Phillips, T., Davydov, E., Golovin, D., et~al.
\newblock Ad click prediction: {A} view from the trenches.
\newblock In \emph{Proceedings of the 19th ACM SIGKDD International Conference
  on Knowledge Discovery and Data Mining}, pp.\  1222--1230. ACM, 2013.

\bibitem[Mikolov et~al.(2013)Mikolov, Chen, Corrado, and
  Dean]{mikolov2013efficient}
Mikolov, T., Chen, K., Corrado, G., and Dean, J.
\newblock Efficient estimation of word representations in vector space.
\newblock \emph{arXiv preprint arXiv:1301.3781}, 2013.

\bibitem[Shalev-Shwartz \& Tewari(2011)Shalev-Shwartz and
  Tewari]{shalev2011stochastic}
Shalev-Shwartz, S. and Tewari, A.
\newblock Stochastic methods for {$\ell_1$}-regularized loss minimization.
\newblock \emph{Journal of Machine Learning Research}, 12\penalty0
  (6):\penalty0 1865--1892, 2011.

\bibitem[Tan et~al.(2014)Tan, Tsang, and Wang]{tan2014towards}
Tan, M., Tsang, I.~W., and Wang, L.
\newblock Towards ultrahigh dimensional feature selection for big data.
\newblock \emph{Journal of Machine Learning Research}, 15\penalty0
  (1):\penalty0 1371--1429, 2014.

\bibitem[Vervier et~al.(2016)Vervier, Mah{\'e}, Tournoud, Veyrieras, and
  Vert]{vervier2016large}
Vervier, K., Mah{\'e}, P., Tournoud, M., Veyrieras, J.-B., and Vert, J.-P.
\newblock Large-scale machine learning for metagenomics sequence
  classification.
\newblock \emph{Bioinformatics}, 32\penalty0 (7):\penalty0 1023--1032, 2016.

\bibitem[Weinberger et~al.(2009)Weinberger, Dasgupta, Langford, Smola, and
  Attenberg]{weinberger2009feature}
Weinberger, K., Dasgupta, A., Langford, J., Smola, A., and Attenberg, J.
\newblock Feature hashing for large scale multitask learning.
\newblock In \emph{Proceedings of the 26th Annual International Conference on
  Machine Learning}, pp.\  1113--1120. ACM, 2009.

\bibitem[Wood \& Salzberg(2014)Wood and Salzberg]{wood2014kraken}
Wood, D.~E. and Salzberg, S.~L.
\newblock Kraken: {U}ltrafast metagenomic sequence classification using exact
  alignments.
\newblock \emph{Genome Biology}, 15\penalty0 (3):\penalty0 1, 2014.

\end{thebibliography}
\bibliographystyle{icml2018}
\clearpage

\section{Appenix}

\newcommand{\ve}{\varepsilon}
\newcommand{\vphi}{\varphi}
In this appendix, we present some preliminary results on the convergence of \mission{}. For the sake of exposition, we will consider the full-gradient descent version of \mission{}, and we will prove that the iterates converge geometrically upto a small additive error. In order to establish this proof, we make an assumption (Assumption~\ref{ass:1}) about the hashing scheme; see Section~\ref{sec:assumption} for more on this.  

We begin by establishing some notation. We will assume that the data satisfies the following linear model: 
\begin{equation}
	y = X \beta^\ast + w, 
\end{equation}
where $y\in \mathbb{R}^n$ is the vector of observation, $X\in \mathbb{R}^{n\times p}$ is the data matrix, $w\in \mathbb{R}^n$ is the noise vector, and $\beta^\ast\in \mathbb{R}^p$ is the unknown $k-$sparse regression vector. We will let $\psi$ and $\vphi$ respectively denote the hashing and the (top-$k$) heavy-hitters operation.  We will let $\beta^t$ denote the output of \mission{} in step $t$. In general, we will let the vector $h\in \mathbb{R}^m$ denote the hash table. Finally, as before, we will let $H_k$ denote the projection operation onto the set of all $k-$sparse vectors. We will make the following assumption about the hashing mechanism: 
\begin{assumption}
\label{ass:1}
	For any $h\in \mathbb{R}^m$, there exists an $\beta_h\in \mathbb{R}^p$ such that the following hold
\begin{enumerate}
\item $\psi(\beta_h) = h$, that is, the hash table contents can be set to $h$ by hashing the vector $\beta_h$. 
\item $\left\| \beta_h - H_k(\beta_h) \right\|_2 \leq \varepsilon_1$
\end{enumerate}
\end{assumption}

This assumption requires the hashing algorithm to be such that there exists a nearly sparse vector that can reproduce any state of the hash table exactly. This is reasonable since the hash table is a near optimal ``code'' for  sparse vectors in $\mathbb{R}^p$.  See Section~\ref{sec:assumption} for more on this. 

We will next state a straightforward lemma about the sketching procedure
\begin{lemma}
\label{lemma:count-sketch-works}
	There exist constants $\ve_2, C_1>0$ such that provided that the size $m$ of the hash table satisfies $m\geq C_1k \log^2 p$, the following holds for any $\beta\in \mathbb{R}^p$ with probability at least $1-\delta_1$: 
\begin{align}
\left\| \vphi(\psi(\beta)) - H_k(\beta) \right\|_2 \leq \varepsilon_2
\end{align}
\end{lemma}
This lemma follows directly from the definition of the Count-Sketch, and we will not prove here. 

We next state the main theorem that we will show. 

\begin{theorem}
\label{thm:main}
	For any $\delta\in \left(0,\frac{1}{3}\right)$ and $\rho \in (0,0.5)$, there is a constant $C>0$ such that the following statement holds with probability at least $1 - 3 \delta$ 
\begin{align}
\left\| \beta^{t+1} - \beta^\ast \right\|_2 &\leq 2 \rho \left\| \beta^t - \beta^\ast \right\|_2 + 2\sqrt{\frac{\sigma_w^2 (1+\mu) k \log p}{n}}\nonumber\\
&\qquad\qquad + 2 \ve_1 + 3\ve_2, 
\end{align}
provided that $n > C k \log p$, $m> Ck \log^2 p$,  and that Assumption~\ref{ass:1} holds. 
\end{theorem}

Notice that since $\rho <0.5$, the above theorem guarantees geometric convergence. This implies that the overall error is of the order of the additive constants $\ve_1$ and $\ve_2$.  

Before we prove this theorem, we will collect some lemmas that will help us prove our result. 

\begin{lemma}
	\label{lemma:inner-product-concentration}
Suppose $X\in \mathbb{R}^{n \times p}$ has i.i.d $\mathcal{N}(0,\frac{1}{n})$ entries. Then for constants $\rho, \delta_2>0$, there exists a constant  $C_2(\delta)>0$ such that if $n \geq C_2 k\log p$ such that for any pair of unit-norm $k-$sparse vectors $\beta_1, \beta_2 \in \mathbb{S}^{p-1}$, the following holds with probability at least $1- \delta_2$. 
\begin{equation}
	\left| \langle X \beta_1, X \beta_2\rangle  - \langle \beta_1, \beta_2 \rangle\right| \leq \rho.
\end{equation}
\end{lemma}

\begin{proof}
	Note that $\mathbb{E}[\langle X\beta_1, X \beta_2\rangle] = \langle \beta_1, \beta_2\rangle$. For a fixed pair of $\beta_1, \beta_2$, the proof follows from a standard Chernoff bound argument after observing that $\langle X\beta_1, X\beta_2\rangle$ can be written as a sum of products of independent Gaussian random variables. The rest of the proof follows from a standard covering argument, which gives the requirement on $n$. 
\end{proof}

\begin{lemma}
	\label{lemma:another-inner-product}
Suppose $X$ has i.i.d entries drawn according to $\mathcal{N}(0,n^{-1})$, and $w\sim \mathcal{N}(0, \sigma_w^2I_n)$ is drawn independently of $X$. Then, for any constant $\delta_3>0$, there are constants $C_3, \mu>0$ such that for all unit norm $k-$sparse $\beta\in \mathbb{S}^{p-1}$, the following holds with probability at least $ 1- \delta_3$:
\begin{align}
\langle \beta, X^T w\rangle \leq \sqrt{\frac{\sigma_w^2(1+\mu)k \log p}{n}}
\end{align}
provided $n\geq C_3 k \log p$. 
\end{lemma}

\begin{proof}
Notice that for a fixed $\beta$, $\langle \beta, X^T w\rangle = \langle X\beta, w\rangle$ has the same distribution as $\frac{1}{\sqrt{n}} \left\| w \right\|_2 \langle\beta, w_2\rangle$, where $w_2 \sim \mathcal{N}(0, I_n)$ is independent of $w$. Now, we can use concentration inequalities of chi-squared random variables to show that there is a constant $C_3'>0$
\begin{align}
\mathbb{P}\left[ \left\| w \right\|_2^2 \geq \sigma_w^2(1 + \mu_1) n \right] \leq e^{-C_3'n}.
\end{align}

Similarly, from chi-squared concentration, there is a constant $C_3''>0$
\begin{align}
\mathbb{P}\left[ \left|\langle \beta, w_2\rangle\right|^2  \geq 1 + \mu_2\right] \leq e^{-C_3''}
\end{align}

Now, with a standard covering argument, we know that there is a constant $C_3'''>0$ such that provided $n > C_3''' k \log p$, the following holds for at least $ 1- \delta_3$ for any $k-$sparse $\beta$:
\begin{align*}
\langle \beta, A^T w\rangle &= \langle A \beta, w\rangle\\
&\leq \sqrt{\frac{\sigma_w^2(1+ \mu) n k \log p }{n}}.
\end{align*}

\end{proof}

\paragraph{Proof of Theorem~\ref{thm:main}}
If we let $h^t$ denote the contents of the hash table at round $t$, notice that we have the following: $x^{t+1} = \varphi(h^{t+1})$. The (full gradient descent version of the) \mission{} algorithm proceeds by updating the hash table with hashes of the gradient updates. Therefore, we have the following relationship: 
\begin{align}
h^{t+1} = h^t + \psi\left( \eta X^TX(\beta^\ast - \beta^t) + X^T w\right)\label{eq:mission-hash-update},
\end{align}
where $\beta^t$ is the output of the algorithm at round $t$. Notice that $\beta^t = \vphi(h^t)$. According to Assumption~\ref{ass:1}, we know that there exists a vector $\tilde{\beta}^t$ such that $\psi(\tilde{\beta}^t) = h^t$. We will use this observation next. Notice that the output of round $t+1$ maybe written as follows:
\begin{align*}
\beta^{t+1} &= \varphi\left( h^t   + \psi\left( \eta X^TX(\beta^\ast - \beta^t) + X^T w\right)\right)\\
&= \varphi\left( \psi\left( \tilde{\beta}^t + \eta X^T X (\beta^\ast - \beta^t) + X^T w\right) \right).
\end{align*}

Now, we will estimate how close the output of the algorithm gets to $\beta^\ast$ in round $t+1$ in terms of how close the algorithm got in round $t$. Notice that 
\begin{align}
&\left\| \beta^{t+1} - \beta^\ast \right\|_2\nonumber\\
 &= \left\|  \varphi\left( \psi\left( \tilde{\beta}^t + \eta X^T X (\beta^\ast - \beta^t) + X^T w\right) \right) - \beta^\ast\right\|_2\nonumber\\
&\leq \left\| H_k\left(\tilde{\beta}^t + \eta X^T X (\beta^\ast - \beta^t) + X^T w\right) - \beta^\ast\right\|_2 + \varepsilon_2,\label{eq:proof-first-step}
\end{align}
which follows from Lemma~\ref{lemma:count-sketch-works}. We will next consider the first term from above. For notational ease, we will set 
$\gamma^{t+1} \triangleq\tilde{\beta}^t + \eta X^T X (\beta^\ast - \beta^t) + X^T w$. Observe that $H_k$ is an orthogonal projection operator, and that $\beta^\ast$ is $k-$sparse, therefore we have that 
\begin{align}
\left\| H_k\left(\gamma^{t+1}\right) -  \gamma^{t+1}\right\|_2^2 \leq \left\| \gamma^{t+1}  - \beta^\ast\right\|_2^2.
\end{align} 
Adding and subtracting $\beta^\ast$ on the left side and cancelling out the common terms, we have the following. 

\begin{align}
&\left\| H_k(\gamma^{t+1}) - \beta^\ast \right\|_2^2\nonumber\\
 &\leq 2 \langle H_k(\gamma^{t+1}) - \beta^\ast, \gamma^{t+1} - \beta^\ast\rangle\nonumber\\
&= 2\langle H_k (\gamma^{t+1}) - \beta^\ast, \tilde{\beta}^t + \eta X^T X (\beta^\ast - \beta^t) + X^T w - \beta^\ast  \rangle\nonumber\\
&= 2\langle H_k (\gamma^{t+1}) - \beta^\ast, \beta^t + \eta X^T X (\beta^\ast - \beta^t) + X^T w - \beta^\ast  \rangle\nonumber\\
& \qquad\qquad+ 2\langle  H_k (\gamma^{t+1}) - \beta^\ast, \beta^t - \tilde{\beta}^t \rangle\nonumber\\
&\stackrel{(a)}{\leq} 2\langle H_k (\gamma^{t+1}) - \beta^\ast, \beta^t + \eta X^T X (\beta^\ast - \beta^t) + X^T w - \beta^\ast  \rangle\nonumber\\
&\qquad\qquad + 2\left\|  H_k (\gamma^{t+1}) - \beta^\ast\right\|_2 \left\| \varphi(\psi(\tilde{\beta}^t)) - \tilde{\beta}^t \right\|_2\nonumber\\
&\leq 2\langle H_k (\gamma^{t+1}) - \beta^\ast, \beta^t + \eta X^T X (\beta^\ast - \beta^t) + X^T w - \beta^\ast  \rangle\nonumber\\
& \qquad+ 2\left\|  H_k (\gamma^{t+1}) - \beta^\ast\right\|_2 \left(\left\| H_k(\tilde{\beta}^t) - \tilde{\beta}^t \right\|_2 +\right.\nonumber\\
&\qquad\qquad\qquad\qquad\qquad\qquad\left. \left\| H_k(\tilde{\beta}^t) - \vphi(\psi(\tilde{\beta}^t)) \right\|_2\right)\nonumber\\
&\stackrel{(b)}{\leq} 2\langle H_k (\gamma^{t+1}) - \beta^\ast, \beta^t + \eta X^T X (\beta^\ast - \beta^t) + X^T w - \beta^\ast  \rangle\nonumber\\
&\qquad\qquad + 2\left\|  H_k (\gamma^{t+1}) - \beta^\ast\right\|_2 \left(\ve_1 + \ve_2\right)\label{eq:proof-second-step},
\end{align}
where $(a)$ follows form the Cauchy-Schwarz inequality and from the definition of $\tilde{\beta}^t$, $(b)$ follows from Assumption~\ref{ass:1} and Lemma~\ref{lemma:count-sketch-works}. We will now turn our attention to the first inner-product in \eqref{eq:proof-second-step}. With some rearrangement of terms, one can see that
\begin{align}
\normalsize
&\langle H_k (\gamma^{t+1}) - \beta^\ast, \beta^t + \eta X^T X (\beta^\ast - \beta^t) + X^T w - \beta^\ast\rangle \nonumber\\
&= \langle H_k (\gamma^{t+1}) - \beta^\ast, \beta^t - \beta^\ast\rangle -\eta \langle X\left(H_k (\gamma^{t+1}) - \beta^\ast\right) , \nonumber\\
&\qquad\qquad X(\beta^t - \beta^\ast)\rangle + \eta\langle H_k(\gamma^{t+1}) - \beta^\ast, X^T w \rangle\nonumber\\
&\stackrel{(a)}{\leq} \rho \left\| H_k (\gamma^{t+1}) - \beta^\ast \right\|_2\left\| \beta^t - \beta^\ast \right\|_2\nonumber\\
&\qquad\qquad\qquad\qquad\qquad\qquad + \langle H_k(\gamma^{t+1}) - \beta^\ast, X^T w\rangle\nonumber\\
& \stackrel{(b)}{\leq}  \rho \left\| H_k (\gamma^{t+1}) - \beta^\ast \right\|_2\left\| \beta^t - \beta^\ast \right\|_2 \nonumber\\
&\qquad\qquad+\left\| H_k (\gamma^{t+1}) - \beta^\ast \right\|_2 \sqrt{\frac{\sigma_w^2 (1+\mu) k \log p}{n}}\label{eq:proof-third-step}
\end{align}
where $(a)$ follows from Lemma~\ref{lemma:inner-product-concentration} and setting $\eta = 1$. $(b)$ follows from Lemma~\ref{lemma:another-inner-product}. 

Putting \eqref{eq:proof-second-step} and \eqref{eq:proof-third-step}, we get
\begin{align}
\left\| H_k\left(\gamma^{t+1}\right) -  \gamma^t\right\|_2 &\leq 2\rho \left\| \beta^t - \beta^\ast \right\|_2\nonumber\\
&\; + 2\sqrt{\frac{\sigma_w^2 (1+\mu) k \log p}{n}} + 2\left( \ve_1 + \ve_2 \right).
\end{align}
Putting this together with \eqref{eq:proof-first-step} gives us the desired result. 

\subsection{On Assumption~\ref{ass:1}}
\label{sec:assumption}
In the full-gradient version of the \mission{} algorithm, one might modify the algorithm explicitly to ensure that Assumption~\ref{ass:1}. Towards this end, one would simply ensure that the gradients vector is attenuated on all but its top $k$ entries at each step. 

It is not hard to see that this \emph{clean-up} step will ensure that Assumption~\ref{ass:1} holds and the rest of the proof simply goes through. In \mission{} as presented in the manuscript, we employ stochastic gradient descent (SGD). While the above proof needs to be modified for it to be applicable to this case, our simulations suggest that this clean-up step is unnecessary here. We suspect that this is due to random cancellations that are introduced by the SGD. This is indeed an exciting avenue for future work.

\end{document}